\documentclass[a4paper,11pt]{article}

\usepackage[utf8]{inputenc}  
\usepackage[T1]{fontenc}     
\usepackage{lmodern}         
\usepackage{geometry}        
\usepackage{amsmath, amssymb, amsthm} 

\usepackage{graphicx}        
\usepackage{xcolor}          

\usepackage{hyperref}        
\hypersetup{
    colorlinks=true,
    linkcolor=blue,
    urlcolor=blue,
    citecolor=red
}

\usepackage{algorithm}
\usepackage{algpseudocode}

\usepackage{tikz}
\usepackage{subcaption}

\usepackage{multicol}
\usepackage{booktabs}

\newcommand{\dft}[1]{\emph{\textbf{#1}}}

\newcommand{\set}[1]{\left\{#1\right\}}
\newcommand{\sucht}{\,\middle |\,}

\DeclareMathOperator{\disj}{disj}

\newcommand{\calA}{\mathcal{A}}

\definecolor{KleinBlue}{RGB}{0, 47, 167}
\colorlet{BaseBlue}{KleinBlue}

\colorlet{AnalogousGreen}{BaseBlue>wheel,5,6}
\colorlet{AnalogousViolet}{BaseBlue>wheel,1,6}

\colorlet{TriadRed}{BaseBlue>wheel,4,12}
\colorlet{TriadGreen}{BaseBlue>wheel,8,12}

\colorlet{TetradMagenta}{BaseBlue>wheel,1,4}
\colorlet{TetradGreen}{BaseBlue>wheel,3,4}

\colorlet{ComplementaryBrown}{BaseBlue>wheel,1,2}

\convertcolorspec{named}{BaseBlue}{rgb}\basebluergb
\convertcolorspec{named}{AnalogousGreen}{rgb}\analogousgreenrgb
\convertcolorspec{named}{AnalogousViolet}{rgb}\analogousvioletrgb
\convertcolorspec{named}{TriadRed}{rgb}\triadredrgb
\convertcolorspec{named}{TriadGreen}{rgb}\triadgreenrgb
\convertcolorspec{named}{TetradMagenta}{rgb}\tetradmagentarbg
\convertcolorspec{named}{TetradGreen}{rgb}\tetradgreenrgb
\convertcolorspec{named}{ComplementaryBrown}{rgb}\complementarybrownrgb

\definecolor{rb-blue}{rgb}{\basebluergb}
\colorlet{rb-blue-light}{rb-blue!10!white}
\colorlet{rb-blue-pastel}{rb-blue!50!white}

\definecolor{rb-red}{rgb}{\triadredrgb}
\colorlet{rb-red-light}{rb-red!10!white}
\colorlet{rb-red-pastel}{rb-red!50!white}
\colorlet{rb-red-bright}{rb-red!50!red}

\definecolor{rb-green}{rgb}{\triadgreenrgb}
\colorlet{rb-green-light}{rb-green!10!white}
\colorlet{rb-green-pastel}{rb-green!50!white}

\definecolor{rb-brown}{rgb}{\complementarybrownrgb}
\colorlet{rb-brown-pastel}{rb-brown!50!white}
\colorlet{rb-brown-light}{rb-brown!15!white}
\colorlet{rb-brown-dark}{rb-brown!20!black}

\definecolor{rb-magenta}{rgb}{\tetradmagentarbg}
\colorlet{rb-magenta-light}{rb-magenta!10!white}
\colorlet{rb-magenta-pastel}{rb-magenta!50!white}

\definecolor{rb-violet}{rgb}{\analogousvioletrgb}
\colorlet{rb-violet-light}{rb-violet!10!white}
\colorlet{rb-violet-pastel}{rb-violet!50!white}

\definecolor{rb-mint}{rgb}{\analogousgreenrgb}
\colorlet{rb-mint-light}{rb-mint!10!white}
\colorlet{rb-mint-pastel}{rb-mint!50!white}

\newtheorem{thm}{Theorem}
\newtheorem{prop}{Proposition}[section]
\newtheorem{lem}[prop]{Lemma}

\newtheorem{cor}[prop]{Corollary}
\newtheorem{eg}[prop]{Example}
\newtheorem{rem}[prop]{Remark}

\title{A Quadratic Lower Bound for Stable Roommates Solvability}
\author{Will Rosenbaum\\
\texttt{W.Rosenbaum@liverpool.ac.uk}\\
Department of Computer Science\\
University of Liverpool
}
\date{\today}

\begin{document}

\maketitle

\begin{abstract}
    In their seminal work on the Stable Marriage Problem (SM), Gale and Shapley introduced a generalization of SM referred to as the Stable Roommates Problem (SR). An instance of SR consists of a set of $2n$ agents, and each agent has preferences in the form of a ranked list of all other agents. The goal is to find a one-to-one matching between the agents that is stable in the sense that no pair of agents have a mutual incentive to deviate from the matching. Unlike the (bipartite) stable marriage problem, in SR, stable matchings need not exist. Irving devised an algorithm that finds a stable matching or reports that none exists in $O(n^2)$ time. In their influential 1989 text, Gusfield and Irving posed the question of whether $\Omega(n^2)$ time is required for SR solvability---the task of deciding if an SR instance admits a stable matching.

    In this paper we provide an affirmative answer to Gusfield and Irving's question. We show that any (randomized) algorithm that decides SR solvability requires $\Omega(n^2)$ adaptive Boolean queries to the agents' preferences (in expectation). Our argument follows from a reduction from the communication complexity of the set disjointness function. The query lower bound implies quadratic time lower bounds for Turing machines, and memory access lower bounds for random access machines. Thus, we establish that Irving's algorithm is optimal (up to a logarithmic factor) in a very strong sense.
\end{abstract}

\section{Introduction}

Since its formalization by Gale and Shapley in 1962~\cite{GS62}, the Stable Marriage Problem (SM) and its variants has become a central topic in economics, computer science, and mathematics. The rich theoretical landscape of stable allocation problems and their diverse applications gained further widespread recognition with the 2012 Nobel Prize in Economics, which was awarded to Roth and Shapley for their contributions to this field. In Gale and Shapley's original formulation of the problem, an SM instance consists of two disjoint sets of agents, traditionally referred to as men and women. Each agent has preferences in the form of a ranked list of agents of the opposite gender. The goal is to find a marriage (i.e., a one-to-one correspondence) between men and women that is stable in the sense that no man-woman pair have a mutual incentive to deviate from the marriage. The seminal result of Gale and Shapley proves that a stable marriage always exists and describes an efficient algorithm for finding one.

In~\cite{GS62}, Gale and Shapley also describe a generalization of SM in which there is no bipartition of the agent set into men and women. In this variant, each agent ranks all other agents---not just those of the opposite gender. Again, the goal is to find a stable matching: a one-to-one correspondence between agents such that no pair has a mutual incentive to deviate from their assigned partners (see Section~\ref{sec:sr} for formal definitions). This generalization of SM is referred to as the Stable Roommates Problem (SR).

Since its introduction by Gale and Shapley, SR has been the subject of a large volume of theoretical work and SR has found many applications. As the problem's name suggests, SR naturally models situations in which ``peers'' should be matched with one another, such as identifying roommates in housing searches~\cite{Perach2008-stable}. Variants of SR can be used to model pair-wise kidney exchanges for organ donation~\cite{Roth2005-pairwise}. The general study of this problem has led to the development of algorithms to facilitate kidney exchange markets for organ donation in several countries~\cite{PairedDonation, LivingKidneyDonation, deKlerka2005-dutch}. SR algorithms have also been employed for peer-to-peer file sharing networks~\cite{Lebedev2007-p2p, Mathieu2010-acyclic}, and forming pairs of players in chess tournaments~\cite{Kujansuu1999-stable}. We refer the readers to~\cite{Manlove2013-algorithmics} for a thorough overview of SR and applications.

Unlike the bipartite SM---where stable marriages are guaranteed to exist for all preferences---Gale and Shapley observed that SR instances need not admit stable matchings (see Section~\ref{sec:unsolvable} for an explicit example). The question of whether or not a given instance admits a stable matching is known as the SR solvability problem. Gale and Shapley~\cite{GS62} and subsequently Knuth~\cite{Knuth76} asked whether there exists an efficient algorithm for SR solvability (and finding a stable matching if one exists), where Knuth suggested the problem might be NP-complete. In 1985, Irving~\cite{Irving1985-stable} devised an efficient algorithm for SR solvability. Irving's algorithm finds a stable matching or reports that none exists in $O(n^2)$ time on instances with $2n$ agents.\footnote{Observe that since the input consists of $2n$ lists each of length $2n - 1$, $O(n^2)$ is \emph{linear} in the input size.} In 1990, Ng and Hirschberg~\cite{HN90} showed that $\Omega(n^2)$ time is necessary to find a stable marriage (in the bipartite stable marriage problem). Their result implies the same lower bound holds for \emph{finding} a stable matching for SR, but the lower bound does not apply to the decision problem of determining whether or not an SR instance admits a stable matching. In their influential 1989 text on stable matching problems, Gusfield and Irving~\cite{GI89} listed finding a lower bound (or $o(n^2)$ time algorithm) for deciding SR solvability as one of 12 open questions for future research. This question was again listed as an open problem in Manlove's comprehensive 2012 text on algorithmic aspects of stable matchings~\cite{Manlove2013-algorithmics}. To our knowledge, the question has remained open since.

\subsection{Our Contributions}

 Our main result is to prove an $\Omega(n^2)$ lower bound for any (randomized or deterministic) algorithm deciding SR solvability. While Gusfield and Irving's question is phrased in terms of running time lower bounds for a particular representation of the agents' preferences, we give a more general result that bounds the number of Boolean queries necessary to decide SR solvability.

\begin{thm}\label{thm:informal}[Informal, c.f.\ Theorem~\ref{thm:main-lb}]
    Any algorithm that decides SR solvability requires $\Omega(n^2)$ Boolean queries to the agents' preferences for instances with $2n$ agents. This lower bound applies to randomized protocols (in expectation) and allows for arbitrary Boolean queries made to individual agents' preferences as well as queries made to predetermined batches of agents (i.e., queries that involve more than one agent's preferences).
\end{thm}

One advantage of phrasing our lower bound in terms of (arbitrary) Boolean queries is that the lower bound extends to many models of computation simultaneously and it is agnostic to the representation of the preferences.

\begin{cor}\label{cor:lb}
    The following lower bounds hold for deciding SR solvability of instances with $2n$ agents:
    \begin{enumerate}
        \item Any (multi-tape, probabilistic) Turing machine that decides SR solvability requires $\Omega(n^2)$ time (in expectation).
        \item Any (randomized) random access machine (RAM) with word size $O(\log n)$ bits that decides SR solvability requires $\Omega(n^2 / \log n)$ memory accesses (in expectation).
    \end{enumerate}
    The lower bounds of~1 and~2 hold even if the agents' preferences are preprocessed arbitrarily in batches of size up to $n/2$ (where separate batches are preprocessed independently of each other).
\end{cor}

Another interpretation of our lower bound is for mechanisms that decide SR solvability by eliciting agents' \emph{implicit} preferences. In practice, agents' preferences may not be known explicitly, even to the agents themselves.\footnote{For example, consider the eponymous application of determining suitable roommates for housing. The question of determining if $a$ prefers potential roommate $b$ to $c$ may rely on knowing a significant amount of information about $b$ and $c$ that is not initially known to $a$. In cases of assigning university housing (where potential roommates may not know each other beforehand) responses to questionnaires are often used as a proxy for explicit preferences, with the hope that responses to the questions elicit sufficient information about the ``true'' preferences to determine roommate compatibility.} Thus, it is desirable to design a mechanism (such as questionnaires) that elicits sufficient information about the agents' preferences to determine SR solvability. Theorem~\ref{thm:informal} implies that any such mechanism requires $\Omega(n^2)$ Boolean queries. Thus, determining SR solvability is essentially as hard as learning the agents' preferences.

In order to prove Theorem~\ref{thm:informal}, we employ a reduction from the two-party communication complexity of the set disjointness function, $\disj$. This technique was previously applied by Gonczarowski et al.~\cite{Gonczarowski2019-stable} to obtain lower bounds for the (bipartite) SM. Given the framework of query lower bounds via communication complexity~\cite{Blais2012-property,Eden2018-lower}, our construction and argument are quite simple. The basic idea is to define a family of SR instances that correspond to inputs to the disjointness function. That is, given an input $(x, y)$ to $\disj$, we define an SR instance $R(x, y)$ such that $R(x, y)$ admits a stable matching if and only if $\disj(x, y) = 1$. We then argue that the correspondence $(x, y) \mapsto R(x, y)$ has the property that any algorithm that decides SR solvability using $q$ queries implies the existence of a communication protocol for $\disj$ using $q$ bits of communication. Our main result follows from well-known communication complexity lower bounds for $\disj$~\cite{KS92,Razborov92}. We describe this technique more thoroughly in Section~\ref{sec:communication} before defining the reduction formally in Section~\ref{sec:lb}.

\begin{rem}
    Our $\Omega(n^2)$ query lower bound is only tight to Irving's algorithm~\cite{Irving1985-stable} up to a logarithmic factor. This is because an SR instance of size $2n$ contains $2n$ lists of $2n - 1$ agents. If each agent's identity is encoded with $\Theta(\log n)$ bits, the total size of preference lists is $\Theta(n^2 \log n)$ bits. Irving's algorithm allows for access to a single entry on a preference list as a unit cost operation, whereas this operation would correspond to $\Theta(\log n)$ Boolean queries in our model. Thus, in our model, Irving's algorithm uses $\Theta(n^2 \log n)$ Boolean queries. We leave it as an open question whether the lower bound can be improved to $\Omega(n^2 \log n)$ Boolean queries (i.e., a lower bound that is truly linear in the instance size) or if the logarithmic factor in the running time can be improved.
\end{rem}

\subsection{Related Work}

SR was introduced by Gale and Shapley~\cite{GS62} who showed that SR instances need not admit stable matchings. Knuth~\cite{Knuth76} explicitly posed the question of whether there is an efficient algorithm to find a stable matching or report that none exists. Irving~\cite{Irving1985-stable} devised an algorithm that finds a stable matching or correctly determines that none exists in $O(n^2)$ time (assuming each preference list entry can be accessed in $O(1)$ time). Gusfield and Irving's influential book~\cite{GI89} gives a comprehensive overview of early algorithmic work on SR, including structural aspects of the set of stable matchings for an SR instance. Additionally,~\cite{GI89} listed two questions related to efficiently deciding SR solvability. The first question asked if it is possible to generate a ``succinct certificate'' for non-solvability of an SR instance. This question was positively answered by Tan~\cite{Tan1991-necessary} who demonstrated that a ``stable partition'' (of size $O(n)$) can be computed in $O(n^2)$ time. Given the partition, solvability of an SR instance can be verified in $O(n)$ time. Gusfield and Irving's second question asked if it is possible to decide SR solvability in $o(n^2)$ time directly. Our lower bound gives a negative answer to this question. Given Tan's result, our lower bound also implies that finding a stable partition requires $\Omega(n^2)$ time. Several previous works~\cite{HN90,CL10,Segal03,Gonczarowski2019-stable} proved $\Omega(n^2)$ lower bounds for \emph{finding} stable marriages in SM (and \emph{a fortiori} for finding stable matchings in SR), but none of these results imply lower bounds for the decision problem of SR solvability. We refer the reader to the text of Manlove~\cite{Manlove2013-algorithmics} for an account of more recent developments related to SR.

Our main lower bound result follows from a reduction from the two-party communication complexity of the disjointness function (see Section~\ref{sec:communication} for a formal definition). The two-party communication complexity model was introduced by Yao in~\cite{Yao79}. The (randomized) two-party communication complexity of the disjointness function was first characterized by Kalyanasundaram and Schintger~\cite{KS92}, and subsequently a simpler argument was found by Razborov~\cite{Razborov92}. 

Most closely related to the present paper is the work of Gonczarowski et al.~\cite{Gonczarowski2019-stable}. In~\cite{Gonczarowski2019-stable}, the authors apply the communication complexity of set disjointness to obtain lower bounds for finding and verifying (bipartite) stable marriages. Indeed, our techniques in this paper closely follow that of~\cite{Gonczarowski2019-stable}. Reductions from communication complexity to obtain (sublinear time) lower bounds were employed by~\cite{Blais2012-property} for property testing. The technique was further codified in the work of Eden and Rosenbaum~\cite{Eden2018-lower} who proved a variety of lower bounds for sublinear time algorithms. Our description of our lower bound in terms of an ``embedding'' of set disjointness follows the general approach and vocabulary of~\cite{Eden2018-lower}.
\section{Background and Definitions}

\subsection{The Stable Roommates Problem}\label{sec:sr}

Formally, an instance of the Stable Roommates problem (hereafter, SR), consists of a set $S$ of $2n$ \dft{agents} together with \dft{preferences} for each agent in the form a ranked list of all other agents. We say that $a$ \dft{prefers} agent $b$ to $c$ if $b$ appears before $c$ on $a$'s preference list. A matching $M = \set{\set{a_1, b_1}, \set{a_2, b_2}, \ldots, \set{a_n, b_n}}$ is a partition of $S$ into subsets of size $2$. Given an SR instance and a matching $M$, we say that a pair $\set{a_i, b_j} \notin M$ is a \dft{blocking pair} if $a_i$ prefers $b_j$ to $b_i$ and $b_j$ prefers $a_i$ to $a_j$. A matching that does not induce any blocking pairs is said to be \dft{stable}.

An SR instance is \dft{solvable} if it admits a stable matching. The \dft{SR Solvability} problem is the problem of deciding if a given SR instance is solvable.

An efficient ($O(n^2)$ time) algorithm for solving SR Solvability was proposed by Irving in~\cite{Irving1985-stable}. We will not require a description of the full algorithm, but our analysis will rely on an initial processing of the preference lists from Irving's algorithm. Our terminology and exposition of Irving's algorithm follow Gusfield and Irving's text~\cite[Chapter~4]{GI89}. 

The algorithm maintains a \dft{preference table} that initially lists all of the agents' preferences. As the algorithm proceeds, pairs of agents that cannot appear in any stable matching are removed. When a pair $\set{a, b}$ is removed from the table, $a$ is removed from $b$'s preference list and $b$ is removed from $a$'s preference list, while the relative order of other agents in the table are unchanged. The initial phase of Irving's algorithm, which we refer to as the Phase~1 algorithm (Algorithm~\ref{alg:phase-1}) proceeds as follows. Each agent is initially ``free.'' Free agents are selected in an arbitrary order to propose to the next remaining agent on their preference list. When $a$ proposes to $b$, $a$ becomes \dft{semiengaged} to $b$ and $a$ is no longer free. Upon receiving a proposal from $a$, $b$ rejects any agent currently semiengaged to $b$, and all pairs of agents $\set{a', b}$ such that $b$ prefers $a$ to $a'$ are removed from the preference table. This process continues until either all agents are semiengaged or all free agents have empty preference lists.

\begin{algorithm}
    \caption{Phase 1 Algorithm~\cite{Irving1985-stable,GI89} \label{alg:phase-1}}
    \begin{algorithmic}[1]
        \Procedure{PhaseOne}{}
        \State Assign each agent to be free
        \While{Some free agent $x$ has a nonempty list}
        \State $y \gets$ first agent on $x$'x preference list \Comment{$x$ proposes to $y$}
        \If{some $z$ is semiengaged to $y$}
        \State assign $z$ to be free \Comment{$y$ rejects $z$}
        \EndIf
        \State assign $x$ to be semiengaged to $y$
        \ForAll{successors $x'$ of $x$ on $y$'s list}
        \State remove $\set{x', y}$ from the preference table \label{ln:remove}
        \EndFor
        \EndWhile
        \EndProcedure
    \end{algorithmic}
\end{algorithm}

Irving showed that when the Phase~1 algorithm terminates, the resulting preference table has the following properties.

\begin{lem}[{Irving~\cite{Irving1985-stable}, c.f.~\cite[Section 4.2.2]{GI89}}]\label{lem:phase-1}
    Consider the preference table resulting from an execution of Algorithm~\ref{alg:phase-1} on an SR instance. Then:
    \begin{enumerate}
        \item The resulting preference table is independent of the order in which free agents propose (\cite[Lemma 4.2.1]{GI89}).
        \item Any pair $\set{a, b}$ that is not contained in the table is not contained in any stable matching. In particular, if any preference list in the table is empty, then the instance is not solvable (\cite[Lemma~4.2.3]{GI89})
        \item If $M$ is a matching such that all pairs in $M$ appear in the preference table, then no pair $\set{a, b}$ not appearing the preference table can block $M$. In particular, if the preference table consists only of a matching $M$, $M$ is the unique stable matching for the instance (c.f., \cite[Lemmas~4.2.1 and~4.2.4]{GI89})

    \end{enumerate}
\end{lem}

In our lower bound construction we will use only Lemma~\ref{lem:phase-1} to argue the (non)solvability of the required SR instances.

\subsubsection{An Unsolvable SR Instance}\label{sec:unsolvable}

In order to motivate the construction used to achieve our main result, it is helpful to understand the following non-solvable instance of SR.

\begin{eg}\label{eg:unsolvable}
    Consider agent set $S = \set{a, b, c, d}$ with the following preferences
    \begin{center}
        \begin{tabular}{c|ccc}
            agent & \multicolumn{3}{c}{preferences}\\
            \hline
            $a$ & $b$ & $c$ & $d$\\
            $b$ & $c$ & $a$ & $d$\\
            $c$ & $a$ & $b$ & $d$\\
            $d$ & \multicolumn{3}{c}{arbitrary}
        \end{tabular}
    \end{center}
    To see why this instance is not solvable, observe that agents $a$, $b$, and $c$ form a ``directed triangle'' where they all rank one another in the top two, but the most preferred agents form a cycle $a \to b \to c \to a$. Thus in any matching $M$, $d$'s partner in $M$ will form a blocking pair with one of the other two matched agents. For example, if $\set{a, d} \in M$, then $\set{a, c}$ forms a blocking pair.
\end{eg}

The idea of our lower bound construction is to define a family of SR instances in which there are many possible sub-instances homomorphic to Example~\ref{eg:unsolvable}. For this family, determining SR solvability is equivalent to finding if a given instance contains such a sub-instance. We argue that determining the existence of such a sub-instance requires $\Omega(n^2)$ accesses to the agent's preferences using a reduction from communication complexity.

\subsection{Communication Complexity and Embeddings}
\label{sec:communication}

Our proof of Theorem~\ref{thm:informal} employs a reduction from two-party communication complexity~\cite{Yao79}. In this computational model, the goal is to compute a value $f(x, y)$ where $f$ is some Boolean function, and $x, y \in \set{0, 1}^N$ are inputs. The inputs are distributed between two parties, traditionally referred to as Alice and Bob. Alice knows only $x$, while Bob knows only $y$. The goal is for Alice and Bob to both learn the value of $f(x, y)$ while exchanging the fewest possible number of bits. The (deterministic) \dft{communication complexity} of $f$ is defined to be the minimum number of bits needed for the worst-case input for any communication protocol that computes $f$. Randomized communication complexity is defined analogously using randomized communication protocols that are allowed to err with some small probability. We refer the reader to the text of Kushilevitz and Nisan~\cite{KN97} for a formal description of the communication models and a thorough exposition of fundamental results.

In this paper, we focus on the communication complexity of the \dft{disjointness function} defined by
\begin{equation}
    \disj(x, y) = 
    \begin{cases}
        1 &\text{ if } \sum_{i = 1}^N x_i y_i = 0\\
        0 &\text{ otherwise.}
    \end{cases}
\end{equation}

The (randomized) communication complexity of the disjointness function was first characterized by Kalyanasundaram and Schintger~\cite{KS92}, and a simpler argument was found by Razborov~\cite{Razborov92} soon after. 

\begin{thm}[{\cite{KS92,Razborov92}}]\label{thm:disjointness}
    Any (deterministic or randomized) communication protocol for $\disj$ requires $\Omega(N)$ bits of communication. This lower bound holds if inputs are restricted to be either disjoint ($\sum_i x_i y_i = 0$) or uniquely intersecting ($\sum_i x_i y_i = 1$).
\end{thm}


Our strategy in obtaining our main result is to define a reduction from $\disj$ to SR solvability. That is, we define a function that maps each input $(x, y)$ for the two party disjointness function to an SR instance $R(x,y)$. This mapping $(x, y) \mapsto R(x, y)$ is an \dft{embedding} of set disjointness in the sense defined by Eden and Rosenbaum~\cite{Eden2018-lower}. This means that the function has the following properties:
\begin{enumerate}
    \item $R(x, y)$ is solvable if and only if $\disj(x, y) = 1$.
    \item The response to any allowed Boolean query to $R(x, y)$ can be computed by Alice (who knows only $x$) and Bob (who knows only $y$) using $O(1)$ bits of communication.
\end{enumerate}

Under a generalization of these conditions, Eden and Rosenbaum~\cite{Eden2018-lower} showed that communication lower bounds imply query lower bounds. We emphasize that this technique yields adaptive query lower bounds: future queries are allowed to depend on the responses to previous queries. Specializing their result to our setting gives the following result.

\begin{thm}[{c.f.~\cite{Eden2018-lower}}]\label{thm:communication-embedding}
    Suppose there exists an embedding from $\disj$ instances of size $N$ to SR solvability instances of size $n$ satisfying the conditions above. Then any randomized or deterministic protocol for SR solvability for instances of size $n$ must use $\Omega(N)$ adaptive queries in expectation.
\end{thm}

We elaborate on what the ``allowed queries'' for our embedding are. In general one cannot allow for arbitrary queries made to arbitrary sets of agents. Indeed, this would allow queries such as ``Is the instance solvable?'' which trivially solves SR solvability with a single query. At the other extreme, one could restrict queries of particular forms such as ``Does agent $a$ prefer agent $b$ to agent $c$?'' Our lower bound holds for arbitrary Boolean queries made to any fixed subsets of agents we call \dft{batches}, so long as each batch contains at most $n/4$ agents. This allows for arbitrary Boolean queries to individual agents' preferences, but also queries whose responses depend on multiple agents within a batch, such as ``Does any agent in batch $A$ prefer agent $b$ to agent $c$?''

\section{The Lower Bound}\label{sec:lb}

In this section we give the main argument for our lower bound for SR solvability. As described above, the idea is to give an embedding of the two party set disjointness problem into SR. The embedding will be as follows. For any positive integer $n$, we will construct SR instances with $4 n$ agents. To this end, let $N = n^2$. To simplify notation in the construction, we index bits of the inputs $x, y$ with two indices from the range $[n]$, i.e., $x = (x_{ij})$ with $1 \leq i, j \leq n$.

Given any pair of Boolean vectors $x, y$, we will define an SR instance $R(x, y)$ as follows. $R(x, y)$ contains $4 n$ agents partitioned into $4$ sets, each of size $n$: $S = A \cup B \cup C \cup D$. We will write the agents in set $A$ as $A = \set{a_1, a_2, \ldots, a_n}$, and similarly with $B$, $C$, and $D$. The preferences of the agents in set $A$ are determined from $x$, and the preferences of agents of $B$ are determined from $y$. The preferences of agents in sets $C$ and $D$ are fixed independent of $x$ and $y$. The preferences are determined as follows.

\begin{figure}[p]
    \centering
    \begin{subfigure}{0.45\textwidth}
        \begin{center}
            \begin{tikzpicture}
                \tikzstyle{uniformNode} = [thick, circle, minimum size=0.75cm, inner sep=0pt, line width=1.5pt]
                \foreach \i in {1,2,3} {
                    \node[draw=rb-blue, style=uniformNode] (a\i) at (0, -\i) {$a_\i$};
                    \node[draw=rb-red, style=uniformNode] (b\i) at (4, -\i) {$b_\i$};
                    \node[draw=rb-blue-pastel, style=uniformNode] (c\i) at (0, -\i-4) {$c_\i$};
                    \node[draw=rb-violet-pastel, style=uniformNode] (d\i) at (4, -\i-4) {$d_\i$};
                }
                \foreach \i in {1,2,3} {
                    \draw[draw=rb-blue-pastel, thick, bend left=45, line width=1.5pt] (c\i) to (a\i);
                    \draw[draw=rb-violet-pastel, thick, bend right=45, line width=1.5pt] (d\i) to (b\i);
                }

                \draw[draw=rb-blue, thick, line width=1.5pt] (a1) to (b1);
                \draw[draw=rb-blue, thick, line width=1.5pt] (a1) to (b2);
                \draw[draw=rb-blue, thick, line width=1.5pt] (a3) to (b2);

                \draw[draw=rb-red, thick, line width=1.5pt] (b2) to (c2);
                \draw[draw=rb-red, thick, line width=1.5pt] (b3) to (c3);

                \draw[draw=rb-red, thick, line width=1.5pt] (b2) to (a2);
                \draw[draw=rb-red, thick, line width=1.5pt] (b3) to (a3);
            \end{tikzpicture}
        \end{center}
        \caption{Initial preferred pairs in the table.}
    \end{subfigure}
    \hfill
    \begin{subfigure}{0.45\textwidth}
        \begin{center}
            \begin{tikzpicture}
                \tikzstyle{uniformNode} = [thick, circle, minimum size=0.75cm, inner sep=0pt, line width=1.5pt]
                \foreach \i in {1,2,3} {
                    \node[draw=rb-blue, style=uniformNode] (a\i) at (0, -\i) {$a_\i$};
                    \node[draw=rb-red, style=uniformNode] (b\i) at (4, -\i) {$b_\i$};
                    \node[draw=rb-blue-pastel, style=uniformNode] (c\i) at (0, -\i-4) {$c_\i$};
                    \node[draw=rb-violet-pastel, style=uniformNode] (d\i) at (4, -\i-4) {$d_\i$};
                }
                \foreach \i in {1,2,3} {
                    \draw[draw=rb-blue-pastel, thick, bend left=45, line width=1.5pt] (c\i) to (a\i);
                    \draw[draw=rb-violet-pastel, thick, bend right=45, line width=1.5pt] (d\i) to (b\i);
                }

                \draw[draw=rb-blue, thick, line width=1.5pt] (a1) to (b1);
                \draw[draw=rb-blue, thick, line width=1.5pt] (a1) to (b2);
                \draw[draw=rb-blue, thick, line width=1.5pt] (a3) to (b2);

                \draw[draw=rb-red, thick, line width=1.5pt] (b2) to (c2);
                \draw[draw=rb-red, thick, line width=1.5pt] (b3) to (c3);

            \end{tikzpicture}
        \end{center}
        \caption{Pairs removed after proposals from $C$.}
    \end{subfigure}\\

    \vspace{1cm}
    
    \begin{subfigure}{0.45\textwidth}
        \begin{center}
            \begin{tikzpicture}
                \tikzstyle{uniformNode} = [thick, circle, minimum size=0.75cm, inner sep=0pt, line width=1.5pt]
                \foreach \i in {1,2,3} {
                    \node[draw=rb-blue, style=uniformNode] (a\i) at (0, -\i) {$a_\i$};
                    \node[draw=rb-red, style=uniformNode] (b\i) at (4, -\i) {$b_\i$};
                    \node[draw=rb-blue-pastel, style=uniformNode] (c\i) at (0, -\i-4) {$c_\i$};
                    \node[draw=rb-violet-pastel, style=uniformNode] (d\i) at (4, -\i-4) {$d_\i$};
                }
                \foreach \i in {1,2,3} {
                    \draw[draw=rb-blue-pastel, thick, bend left=45, line width=1.5pt] (c\i) to (a\i);
                    \draw[draw=rb-violet-pastel, thick, bend right=45, line width=1.5pt] (d\i) to (b\i);
                }


                \draw[draw=rb-red, thick, line width=1.5pt] (b2) to (c2);
                \draw[draw=rb-red, thick, line width=1.5pt] (b3) to (c3);

            \end{tikzpicture}
        \end{center}
        \caption{Pairs removed after proposals from $D$.}
    \end{subfigure}
    \hfill
    \begin{subfigure}{0.45\textwidth}
        \begin{center}
            \begin{tikzpicture}
                \tikzstyle{uniformNode} = [thick, circle, minimum size=0.75cm, inner sep=0pt, line width=1.5pt]
                \foreach \i in {1,2,3} {
                    \node[draw=rb-blue, style=uniformNode] (a\i) at (0, -\i) {$a_\i$};
                    \node[draw=rb-red, style=uniformNode] (b\i) at (4, -\i) {$b_\i$};
                    \node[draw=rb-blue-pastel, style=uniformNode] (c\i) at (0, -\i-4) {$c_\i$};
                    \node[draw=rb-violet-pastel, style=uniformNode] (d\i) at (4, -\i-4) {$d_\i$};
                }
                \foreach \i in {1,2,3} {
                    \draw[draw=rb-blue-pastel, thick, bend left=45, line width=1.5pt] (c\i) to (a\i);
                    \draw[draw=rb-violet-pastel, thick, bend right=45, line width=1.5pt] (d\i) to (b\i);
                }



            \end{tikzpicture}
        \end{center}
        \caption{Remaining pairs after proposals from $A$.}
    \end{subfigure}
    \caption{An illustration of the embedding of disjointness for $N = 3 \times 3$ and $n = 3$. This instance corresponds to $x_{1,1} = x_{1,2} = x_{3, 2} = 1$, while the remaining values of $x_{ij} = 0$, and $y_{2,2} = y_{3,3} = 1$ with the remaining values of $y_{ij} = 0$. Thus, $\disj(x, y) = 1.$ The edges between agents are colored by agent preferences: the dark blue edges from the $a_i$ correspond to their most preferred partners (preferred above $c_i$). Similarly, the dark red edges from the $b_j$ correspond to their most preferred edges. The light blue edges from the $c_i$ and violet edges from the $d_j$ correspond to those agent's most preferred partners. Other possible partners are not depicted. Sub-figure~(a) represents the remaining pairs in the preference table before the first rounds of proposals, while (b), (c), and (d) depict the remaining pairs after each round of proposals. The final figure depicts the unique stable matching for this instance.}
    \label{fig:disjoint}
\end{figure}
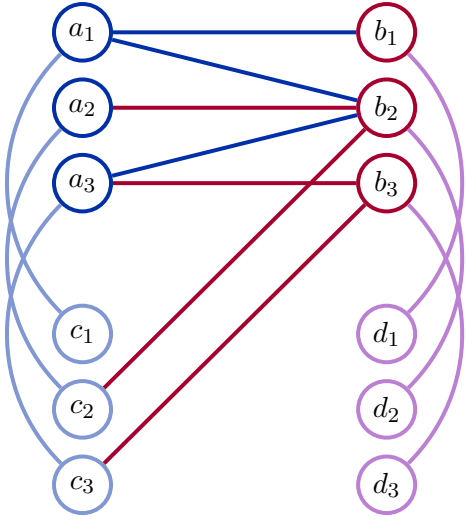
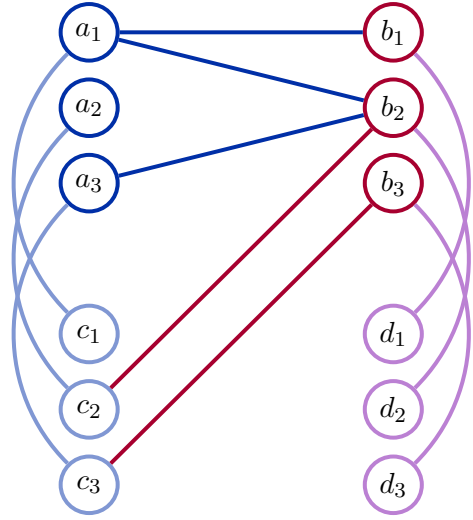
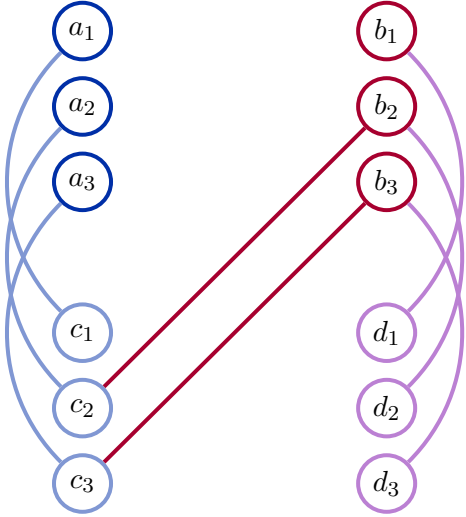
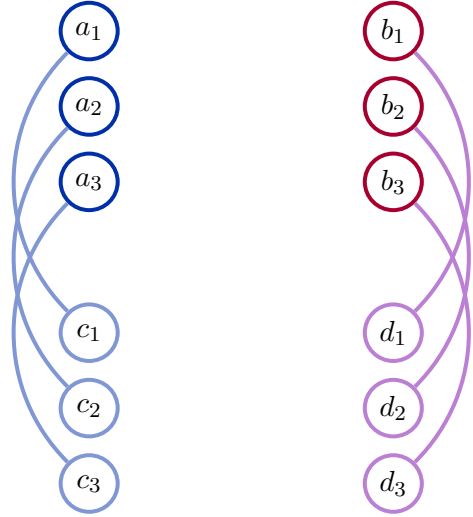

\begin{multicols}{2}
\begin{itemize}
    \item An agent $a_i \in A$ prefers in order
    \begin{enumerate}
        \item all $b_j$ with $x_{i j} = 1$ (arbitrary order)
        \item $c_i$
        \item all other agents in arbitrary order
    \end{enumerate}
    \item An agent $b_j \in B$ prefers in order
    \begin{enumerate}
        \item all $c_i$ with $y_{i j} = 1$ (arbitrary order)
        \item all $a_i$ with $y_{i j} = 1$ (arbitrary order)
        \item $d_j$
        \item all other agents in arbitrary order
    \end{enumerate}
    \columnbreak

    \item An agent $c_i \in C$ prefers in order
    \begin{enumerate}
        \item $a_i$
        \item all agents $b_j \in B$ in arbitrary order
        \item all other agents in arbitrary order
    \end{enumerate}
    \item An agent $d_j \in D$ prefers in order
    \begin{enumerate}
        \item $b_j$
        \item all other agents in arbitrary order
    \end{enumerate}
    \vfill
\end{itemize}
\end{multicols}

With the SR instances $R(x, y)$ defined as above, we state the main property regarding the solvability of $R(x, y)$.

\begin{prop}
    \label{prop:embedding}
    Suppose $x, y \in \set{0, 1}^{n \times n}$, and let $R(x, y)$ be the SR instance defined above. 
    \begin{enumerate}
        \item If $\disj(x, y) = 1$ (i.e., $x$ and $y$ are disjoint), then $R(x, y)$ admits a unique stable matching $M$, namely 
        \begin{equation}
            M = \set{\set{a_i, c_i} \sucht i \in [n]} \cup \set{\set{b_j, d_j} \sucht j \in [n]}
        \end{equation}
        \item If $x$ and $y$ are uniquely intersecting (i.e., $\sum_{i,j} x_{ij} y_{ij} = 1$), and hence $\disj(x, y) = 0$, then $R(x, y)$ is not solvable.
    \end{enumerate}
    In particular, the association $(x, y) \mapsto R(x, y)$ is an embedding of disjointness into SR solvability.
\end{prop}
\begin{proof}
    To prove the proposition, we consider an execution of the Phase~1 algorithm, Algorithm~\ref{alg:phase-1}. We first consider the case where $\disj(x, y) = 1$ (case~1 of the proposition). Consider an execution of Algorithm~\ref{alg:phase-1} with the agents proposing in the following order. By Lemma~\ref{lem:phase-1}, the resulting preference table does not depend on the order of the proposals. We illustrate an example of the execution of such an instance with $n = 3$ in Figure~\ref{fig:disjoint}.
    \begin{enumerate}
        \item Agents in $C$ propose. Specifically, each $c_i \in C$ proposes to $a_i$. Since $a_i$ receives a proposal from $c_i$, in Line~\ref{ln:remove}, only $b_j$s with $x_{ij} = 1$ will remain on $a$'s preference list (above $c_i$). Since $\disj(x, y) = 1$, $x_{ij} = 1$ implies $y_{ij} = 0$, hence all pairs $\set{a_i, b_j}$ with $y_{ij} = 1$ are removed.
        \item Agents in $D$ propose. Specifically, each $d_j \in D$ proposes to $b_j$. When $b_j$ receives the proposal from $d_j$, only $c_i$s with $y_{ij} = 1$ remain on $b_j$'s preference list, as all $a_i$'s with $y_{ij} = 1$ were removed after the proposals in Step~1. Since $\disj(x, y) = 1$, if $y_{ij} = 1$, then $x_{ij} = 0$. Therefore, all remaining pairs of the form $\set{a_i, b_j}$ are removed from the preference table.
        \item Agents in $A$ propose. At this point, each $a_i$'s preference list consists of a single entry, $c_i$, so $a_i$ proposes to $c_i$. Since $a_i$ is first on $c_i$'s preference list, all remaining pairs involving $c_i$ are removed from the preference table. In particular, all pairs of the form $\set{c_i, b_j}$ are removed.
        \item Agents in $B$ propose. At this point, each $b_j$'s preference list consists of a single entry, $d_j$, so $b_j$ proposes to $d_j$.
    \end{enumerate}
    After agents from $B$ propose, all agents are semiengaged so the algorithm terminates. The preference table consists only of the matching $M$ in the proposition's statement, which is therefore a stable matching (by Lemma~\ref{lem:phase-1}).
    
    Now consider the case where $x$ and $y$ are uniquely intersecting. Let $k, \ell \in [n]$ be the indices for which $x_{k \ell} = y_{k \ell} = 1$, while $x_{i j} y_{i j} = 0$ for all $(i, j) \neq (k, \ell)$. As above, we consider an execution of Algorithm~\ref{alg:phase-1}. We illustrate an example of the execution of such an instance in Figure~\ref{fig:intersecting}.
    \begin{enumerate}
        \item Agents in $C$ propose. In this case, pairs $\set{a_i, b_j}$ with $y_{ij} = 1$ are removed from the preference table except for $\set{a_k, b_\ell}$. Further, all pairs $\set{a_i, d_j}$ are removed.
        \item Agents in $D$ propose. Now all pairs $\set{a_i, b_j}$ with $x_{ij} = 1$ are removed from the preference table except for $\set{a_k, b_\ell}$.
        \item Agents in $A$ propose. All $a_{i}$ with $i \neq k$ propose to $c_i$, while $a_k$ proposes to $b_\ell$. Subsequently, $b_\ell$ rejects $d_\ell$. As each agent $c_i$ with $i \neq k$ receives a proposal from $a_i$, each $\set{c_i, d_\ell}$ is removed from the preference table.
        \item Agents in $B$ propose. In this case, each $b_j$ with $j \neq \ell$ proposes to $d_j$, while $b_\ell$ proposes to $c_k$. Subsequently, all pairs $\set{d_j, d_\ell}$ are removed from the preference table, as is the pair $\set{c_k, d_\ell}$.
    \end{enumerate}
    At this point all pairs involving $d_\ell$ have been removed, so $d_\ell$'s preference list is empty. Therefore, by Lemma~\ref{lem:phase-1}, $R(x, y)$ does not admit a stable matching.
\end{proof}

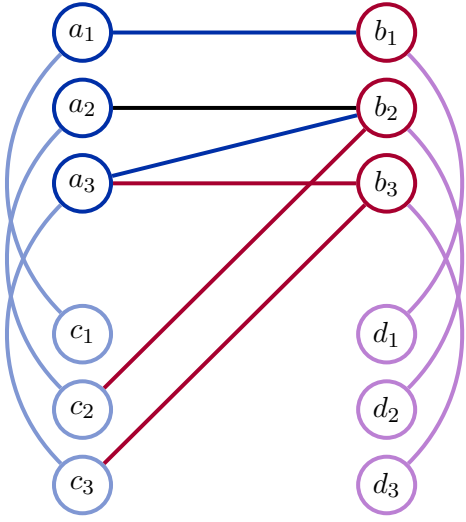
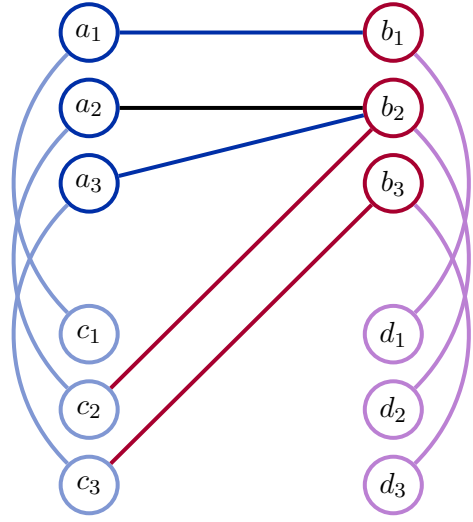
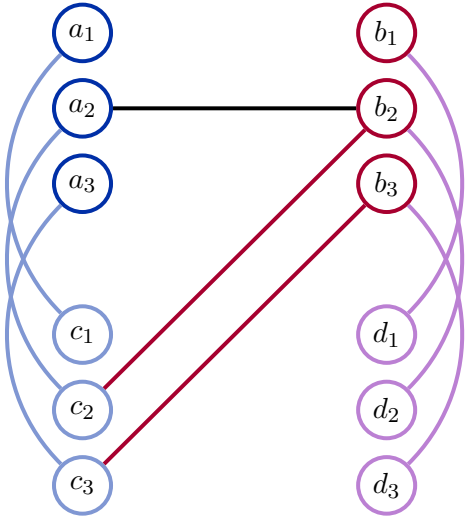
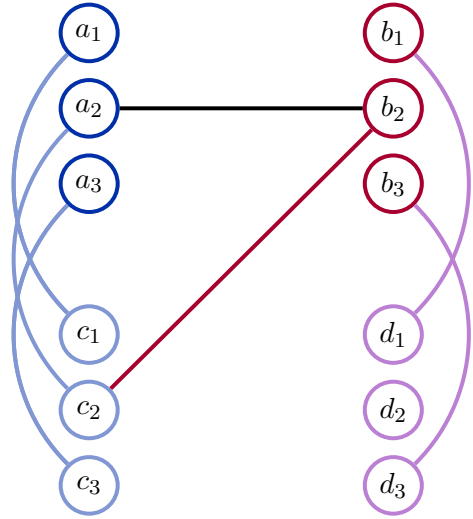
\begin{figure}[p]
    \centering
    \begin{subfigure}{0.45\textwidth}
        \begin{center}
            \begin{tikzpicture}
                \tikzstyle{uniformNode} = [thick, circle, minimum size=0.75cm, inner sep=0pt, line width=1.5pt]
                \foreach \i in {1,2,3} {
                    \node[draw=rb-blue, style=uniformNode] (a\i) at (0, -\i) {$a_\i$};
                    \node[draw=rb-red, style=uniformNode] (b\i) at (4, -\i) {$b_\i$};
                    \node[draw=rb-blue-pastel, style=uniformNode] (c\i) at (0, -\i-4) {$c_\i$};
                    \node[draw=rb-violet-pastel, style=uniformNode] (d\i) at (4, -\i-4) {$d_\i$};
                }
                \foreach \i in {1,2,3} {
                    \draw[draw=rb-blue-pastel, thick, bend left=45, line width=1.5pt] (c\i) to (a\i);
                    \draw[draw=rb-violet-pastel, thick, bend right=45, line width=1.5pt] (d\i) to (b\i);
                }

                \draw[draw=rb-blue, thick, line width=1.5pt] (a1) to (b1);
                \draw[draw=black, thick, line width=1.5pt] (a2) to (b2);
                \draw[draw=rb-blue, thick, line width=1.5pt] (a3) to (b2);

                \draw[draw=rb-red, thick, line width=1.5pt] (b2) to (c2);
                \draw[draw=rb-red, thick, line width=1.5pt] (b3) to (c3);

                \draw[draw=rb-red, thick, line width=1.5pt] (b3) to (a3);
            \end{tikzpicture}
        \end{center}
        \caption{Initial preferred pairs in the table.}
    \end{subfigure}
    \hfill
    \begin{subfigure}{0.45\textwidth}
        \begin{center}
            \begin{tikzpicture}
                \tikzstyle{uniformNode} = [thick, circle, minimum size=0.75cm, inner sep=0pt, line width=1.5pt]
                \foreach \i in {1,2,3} {
                    \node[draw=rb-blue, style=uniformNode] (a\i) at (0, -\i) {$a_\i$};
                    \node[draw=rb-red, style=uniformNode] (b\i) at (4, -\i) {$b_\i$};
                    \node[draw=rb-blue-pastel, style=uniformNode] (c\i) at (0, -\i-4) {$c_\i$};
                    \node[draw=rb-violet-pastel, style=uniformNode] (d\i) at (4, -\i-4) {$d_\i$};
                }
                \foreach \i in {1,2,3} {
                    \draw[draw=rb-blue-pastel, thick, bend left=45, line width=1.5pt] (c\i) to (a\i);
                    \draw[draw=rb-violet-pastel, thick, bend right=45, line width=1.5pt] (d\i) to (b\i);
                }

                \draw[draw=rb-blue, thick, line width=1.5pt] (a1) to (b1);
                \draw[draw=black, thick, line width=1.5pt] (a2) to (b2);
                \draw[draw=rb-blue, thick, line width=1.5pt] (a3) to (b2);

                \draw[draw=rb-red, thick, line width=1.5pt] (b2) to (c2);
                \draw[draw=rb-red, thick, line width=1.5pt] (b3) to (c3);

            \end{tikzpicture}
        \end{center}
        \caption{Pairs removed after proposals from $C$.}
    \end{subfigure}\\

    \vspace{1cm}
    
    \begin{subfigure}{0.45\textwidth}
        \begin{center}
            \begin{tikzpicture}
                \tikzstyle{uniformNode} = [thick, circle, minimum size=0.75cm, inner sep=0pt, line width=1.5pt]
                \foreach \i in {1,2,3} {
                    \node[draw=rb-blue, style=uniformNode] (a\i) at (0, -\i) {$a_\i$};
                    \node[draw=rb-red, style=uniformNode] (b\i) at (4, -\i) {$b_\i$};
                    \node[draw=rb-blue-pastel, style=uniformNode] (c\i) at (0, -\i-4) {$c_\i$};
                    \node[draw=rb-violet-pastel, style=uniformNode] (d\i) at (4, -\i-4) {$d_\i$};
                }
                \foreach \i in {1,2,3} {
                    \draw[draw=rb-blue-pastel, thick, bend left=45, line width=1.5pt] (c\i) to (a\i);
                    \draw[draw=rb-violet-pastel, thick, bend right=45, line width=1.5pt] (d\i) to (b\i);
                }

                \draw[draw=black, thick, line width=1.5pt] (a2) to (b2);

                \draw[draw=rb-red, thick, line width=1.5pt] (b2) to (c2);
                \draw[draw=rb-red, thick, line width=1.5pt] (b3) to (c3);

            \end{tikzpicture}
        \end{center}
        \caption{Pairs removed after proposals from $D$.}
    \end{subfigure}
    \hfill
    \begin{subfigure}{0.45\textwidth}
        \begin{center}
            \begin{tikzpicture}
                \tikzstyle{uniformNode} = [thick, circle, minimum size=0.75cm, inner sep=0pt, line width=1.5pt]
                \foreach \i in {1,2,3} {
                    \node[draw=rb-blue, style=uniformNode] (a\i) at (0, -\i) {$a_\i$};
                    \node[draw=rb-red, style=uniformNode] (b\i) at (4, -\i) {$b_\i$};
                    \node[draw=rb-blue-pastel, style=uniformNode] (c\i) at (0, -\i-4) {$c_\i$};
                    \node[draw=rb-violet-pastel, style=uniformNode] (d\i) at (4, -\i-4) {$d_\i$};
                }
                \foreach \i in {1,2,3} {
                    \draw[draw=rb-blue-pastel, thick, bend left=45, line width=1.5pt] (c\i) to (a\i);
                }                
                \foreach \i in {1,3} {
                    \draw[draw=rb-blue-pastel, thick, bend left=45, line width=1.5pt] (c\i) to (a\i);
                    \draw[draw=rb-violet-pastel, thick, bend right=45, line width=1.5pt] (d\i) to (b\i);
                }

                \draw[draw=black, thick, line width=1.5pt] (a2) to (b2);

                \draw[draw=rb-red, thick, line width=1.5pt] (b2) to (c2);

            \end{tikzpicture}
        \end{center}
        \caption{Remaining pairs after proposals from $A$.}
    \end{subfigure}
    \caption{An illustration of the embedding of disjointness for $N = 3 \times 3$ and $n = 6$. This instance corresponds to $x_{1,1} = x_{2,2} = x_{3, 2} = 1$, while the remaining values of $x_{ij} = 0$, and $y_{2,2} = y_{3,3} = 1$ with the remaining values of $y_{ij} = 0$. Thus, $\disj(x, y) = 1$ with $x_{2,2} = y_{2,2} = 1$. Sub-figure~(a) represents the remaining pairs in the preference table before the first rounds of proposals, while (b), (c), and (d) depict the remaining pairs after each round of proposals. Note that in all of the figures, the pair $\set{a_2, b_2}$ is preferred by $a_2$ to $c_2$ and preferred by $b_2$ to $d_2$. Therefore, this edge is not removed after either the $C$ proposals nor the $D$ proposals in figures~(b) and~(c). Subsequently, when agents in $A$ propose, $a_2$ proposes to $b_2$, after which $b_2$ rejects $d_2$. At this point $d_2$'s preference list is empty, hence the instance does not admit a stable matching.}
    \label{fig:intersecting}
\end{figure}

We now state and prove our main lower bound for the stable roommates problem. In the formal statement, we allow for any procedure that performs arbitrary adaptive Boolean queries to the agents' preference lists, and queries can even be made to ``batches'' of agents---i.e., a fixed partition of the agents---so long as no batch contains more than $n/2$ agents. For example, this allows queries of the form, ``Does any agent in set $A$ prefer agent $b$ to $b'$?'' Specifically, our argument holds for any Boolean query that does not involve agents from \emph{both} sets $A$ and $B$ in the construction above. An upper bound on the batch size is necessary allowing a batch size of $n$ would allow for the query ``Does $S$ admit a stable matching?'' as a single query.

\begin{thm}\label{thm:main-lb}
    Any randomized (or deterministic) mechanism that decides SR solvability on instances with $n$ agents using adaptive Boolean query access to the agents' preferences requires $\Omega(n^2)$ Boolean queries in expectation. This bound applies even if queries can be made to fixed batches of agents of size up to $n/2$.
\end{thm}
\begin{proof}
    Let $\calA$ be any algorithm that determines SR solvability using $q$ queries for SR instances of size $n$. Then we can use $\calA$ to define a two party communication protocol for $\disj$ using the embedding of Proposition~\ref{prop:embedding} as follows. Suppose Alice and Bob hold inputs $x, y \in \set{0,1}^{n \times n}$, respectively, which are promised to be either disjoint or uniquely intersecting. Note that the set disjointness instance has size $N = n^2$. Alice forms preferences of agents in the set $A$ as in the embedding, and Bob does the same for $B$. Thus, the instance has $4n$ agents in total. Note that the preferences of agents in $C$ and $D$ are independent of $x$ and $y$, hence they are known to both Alice and Bob. Alice and Bob then simulate the algorithm $\calA$: the response to any query that $\calA$ makes to agents in $A \cup C \cup D$ can be computed by Alice, who then sends the Boolean result to Bob as a single bit. Similarly, the response to any query made to $B \cup C \cup D$ can be computed by Bob, who then sends the response to Alice. When $\calA$ terminates after $q$ queries, both Alice and Bob can compute the output of $\calA$ from the communication transcript, hence they both know whether or not the $R(x, y)$ is solvable. By Proposition~\ref{prop:embedding}, this output determines $\disj(x, y)$.

    Since this communication protocol solves an arbitrary instance of $\disj$ (with the unique intersection promise) with input size $n^2$ using $q$ bits of communication, we have $q = \Omega(n^2)$ by Theorem~\ref{thm:disjointness}.
\end{proof}

Finally, we argue that the query lower bound of Theorem~\ref{thm:main-lb} implies the computational lower bounds listed in Corollary~\ref{cor:lb}.

\begin{proof}[Proof of Corollary~\ref{cor:lb} (Sketch)]
    The lower bound for Turing machines follows from the observation that each ``read'' operation performed by a Turing machine can be modelled as a response to a single Boolean query (the value of the bit that is read). Thus, the query lower bound implies that any Turing machine the decides SR solvability must read $\Omega(n^2)$ bits of its input, hence its running time is $\Omega(n^2)$. 
    
    Similarly, for random access machines (RAMs), each memory access to a word of size $O(\log n)$ bits can be simulated by $O(\log n)$ Boolean queries. Hence we obtain a $\Omega(n^2/\log n)$ lower bound on the number of memory accesses.

    We observe that the arguments above make no reference to the representation of the input to the problem, and the query lower bound argument holds for any Boolean queries made to the input so long as a single query's response does not depend on the preferences of agents from both $A$ and $B$. Following an arbitrary preprocessing of the agents' preferences where agents in $A$ are processed independently from agents in $B$, each bit of the resulting encoding is determined by either preferences in $A \cup C \cup D$ or preferences in $B \cup C \cup D$. In the former case, reading a single bit of the processed input can be simulated with a single Boolean query to $A$ (as preferences in $C$ and $D$ are fixed) in the former case and $B$ in the latter case. Thus, the lower bounds apply to preprocessed inputs as well.
\end{proof}




\paragraph{Acknowledgements} I am thankful to Christine Cheng for reigniting my interest in the stable roommates problem and to David Manlove for his suggestions to improve this manuscript.

\bibliographystyle{plain}
\bibliography{references}

\end{document}